\documentclass[11pt,twoside]{article}
\usepackage[margin=1in]{geometry}

\usepackage[T1]{fontenc}    
\usepackage{hyperref} 
\usepackage{nicefrac}
\usepackage{url}            
\usepackage{booktabs}       
\usepackage{amsfonts,amsmath}       
\usepackage{nicefrac}       
\usepackage{microtype}      
\usepackage{xcolor}         
\usepackage{amssymb}
\usepackage{amsthm}
\usepackage{algorithm}
\usepackage{multirow}
\usepackage{url}  
\usepackage{verbatim}
\usepackage{color}
\usepackage{graphicx}  
\usepackage{algorithm}
\usepackage{algorithmic}
\newsavebox{\algbox}
\usepackage{mathtools}
\usepackage{caption}
\usepackage{subcaption}
\usepackage{scalerel,stackengine}

\newtheorem{conj}{Conjecture}
\newtheorem{thm}{Theorem}

\newtheorem{lemma}[conj]{Lemma}

\newtheorem{prop}[conj]{Proposition}

\newtheorem{defn}[conj]{Definition}
\newtheorem{coro}{Corollary}

\newcommand{\f}[1]{\boldsymbol{#1}}
\newcommand{\bb}[1]{\mathbb{#1}}
\newcommand{\fl}[1]{\mathbf{#1}}
\newcommand{\ca}[1]{\mathcal{#1}}

\newcommand{\s}[1]{\mathsf{#1}}

\newcommand{\remove}[1]{}

\title{Support Recovery in Universal One-bit Compressed Sensing}

 \author{~~~Arya~Mazumdar\footnote{A.Mazumdar is with the Halicio\v{g}lu Data Science Institute (HDSI) at University of California, San Diego USA (email: \texttt{arya@ucsd.edu}).} ~~~Soumyabrata~Pal\footnote{S. Pal is with the Computer Science Department at the University of Massachusetts Amherst, Amherst, MA 01003, USA (email: \texttt{spal@cs.umass.edu}).}}

\begin{document}

\maketitle

\begin{abstract}
  One-bit compressed sensing (1bCS) is an extreme-quantized signal acquisition method that has been intermittently studied in the past decade. In  1bCS, linear samples of a high dimensional signal are quantized to only one bit per sample (sign of the measurement). The extreme quantization makes it an interesting case study of the more general single-index or generalized linear models. At the same time it can also be thought of as a `design' version of learning a binary linear classifier or halfspace-learning.
  
  Assuming the original signal vector to be sparse, existing results in 1bCS either aim to find the support of the vector, or approximate the signal within an $\epsilon$-ball. The focus of this paper is support recovery, which often also computationally facilitate approximate signal recovery.    A {\em universal} measurement matrix for 1bCS  refers to one set of measurements that work {\em for all} sparse signals. With universality, it is known that $\tilde{\Theta}(k^2)$ 1bCS measurements are necessary and sufficient for support recovery (where $k$ denotes the sparsity). In this work, we show that it is possible to universally recover the support with a small number of false positives with $\tilde{O}(k^{3/2})$ measurements. If the dynamic range of the signal vector is known, then with a different technique, this result can be improved to only $\tilde{O}(k)$ measurements. Other results on universal but approximate support recovery are also provided in this paper. All of our main recovery algorithms are simple and polynomial-time.
\end{abstract}

\section{Introduction}
One-bit compressed sensing (1bCS) is a sampling mechanism for high-dimensional sparse signals, introduced first by Boufounos and Baraniuk \cite{DBLP:conf/ciss/BoufounosB08}. 
The method of obtaining signals by taking few linear projections is known as compressed sensing~\cite{DBLP:journals/tit/Donoho06,candes2006robust}. Given the success of compressed sensing, two points can be noted. First, it is impossible to record real numbers in digital systems without quantization; second, sampling with nonlinear operators can potentially be useful. One-bit compressed sensing is a case-study in both of these fronts.
In terms of quantization, this is the extreme setting where only one bit per sample is acquired. In terms of nonlinearity, this is the one of the simplest example of a single-index model~\cite{plan2016generalized}: $y_i = f(\langle \fl{a}_i, \fl{x} \rangle), i =1, \dots, m$, where $f$ is a coordinate-wise nonlinear operation. In the particular case of $f$ being the $\mathrm{sign}$ function, the model is also the same as that of a simple binary hyperplane classifier. 
For these reasons, 1bCS is also studied with some interest in the last few years, for example, in \cite{DBLP:conf/ciss/HauptB11,GNJN13,ABK17,DBLP:journals/tit/PlanV13,DBLP:conf/aistats/Li16}. 

Most of the existing results either aim for support recovery, or approximate vector recovery for the signal from nonadaptive measurements. It is assumed that the original signal $\fl{x} \in \bb{R}^n$ is $k$-sparse, or has at most $k$ nonzero entries (also written as $\|\fl{x}\|_0 \le k).$ The support recovery results aim to recover the coordinates that have nonzero values; whereas the approximate vector recovery results aim to reconstruct the vector up to some Euclidean distance. It is known that recovering the support can be useful in terms of making the approximate recovery part computationally fast~\cite{GNJN13}.  In this paper we primarily restrict ourselves to support recovery.

A notion that is going to be important moving forward in this paper is that of {\em universality}. A set of measurements (can be stacked in form of a matrix) is called {\em universal} if a recovery guarantee can be given {\em for all} sparse signals. Universal measurements are desirable in any practical application, including hardware design, since one does not have to change the measurement vectors every time for a new signal. Note that, in the canonical works of compressed sensing, the measurement matrices are almost always shown to be universal (i.e., Gaussian or Bernoulli matrices are universal reconstruction matrices with high probability).  

This brings the natural question, how many measurements are necessary and sufficient for support recovery in universal 1bCS? A simple counting bound shows that $\Omega(k \log (n/k))$ measurements are required, where $k$ and $n$ refers to the sparsity and dimension of the signal respectively. This naive bound has been improved recently, and it was shown that, in fact $\Omega(k^2 \log n / \log k)$ measurements are required for universal support recovery~\cite{ABK17}.
What about sufficient number of measurements? Using measurements given by some combinatorial designs, it was shown that $O(k^2 \log n)$ measurements are enough for support recovery~\cite{ABK17}, thereby leaving only a gap of $O(\log k)$-factor between upper and lower bounds.

The price of universality on the other hand is quite steep. Without the requirement that the measurement matrix work for all signals, it turns out that the number of sufficient measurements for support recovery is $O(k \log n)$~\cite{DBLP:conf/ciss/HauptB11}. Therefore, to impose universality, the number of measurements must grow by a factor of $\tilde{\Omega}(k)$. In this work we show that by allowing a few false positives it is possible to substantially bring down this gap. In fact, it is possible to recover entirety of the support with at most $\epsilon k$ false positives, $\epsilon >0$, with only $O(k^{3/2} \log (n/k))$ universal measurements. This result can be improved to $O(k \log (n/k))$ when either a) we allow a few false negatives, or b) we have knowledge about the dynamic range of the signal. This practically cancels the the penalty that one has to pay for universality.

Note that, while allowing few false positives were considered in \cite{flodin2019superset}, their results were only restricted to positive signal vectors, and therefore not truly universal.

\subsection{Key difference from group testing, binary matrices, and technical motivation}\label{sec:differences}
Support recovery in the 1bCS problem has some similarity/connection with the combinatorial {\em group testing} problem~\cite{du2000combinatorial}. In group testing, the original signal $\fl{x}$ is binary (has only $0$s and $1$s), and the measurement matrix has to be binary as well. While in the original compressed sensing problem the main tools are linear algebraic and relate to isometric embeddings, in  group testing most tools are combinatorial and relate to a variety of set systems.

As noted in \cite{ABK17}, group testing and 1bCS have many parallels. Indeed, for universal support recovery, measurement matrices were constructed using union-free set systems, similar to group testing. The upper and lower bound on the number of measurements required for support recovery in 1bCS is also same as group testing (i.e., $O(k^2 \log n)$ and $\Omega(k^2 \log n/ \log k)$). It is therefore believable that by relaxing the recovery condition to allow some false positives, one will obtain an improvement in terms of number of measurements in 1bCS, as in the case of group testing \cite{ngo2011efficiently}. What is more, perhaps support recovery in 1bCS can be performed with a {\em binary} matrix, as in the case of group testing.

Indeed, using a modification of the standard matrices for group testing, as well as using a modified recovery algorithm, Acharya et al.~\cite{ABK17} were able to use $O(k^2 \log n)$ measurements for exact recovery of the support. This is within a $\log k$ factor of the lower bound and achieved with a binary measurement matrix. However, when subsequently  recovery with some false positives were tried~\cite{flodin2019superset}, the group testing performance could not be replicated. In fact, it turned out there were no improvement from the $O(k^2 \log n)$ upper bound in 1bCS if universality is to be preserved.

The main reason why this happens is the following. When a vector $\fl{x}$ is measured with a measurement vector $\fl{a}$ in group testing, an output of $0$ implies that the supports of $\fl{x}$ and $\fl{a}$ do not intersect. Whereas, in 1bCS, it can simply mean that $\fl{x}$ and $\fl{a}$ are orthogonal. To be sure of what the measurement outcome of $0$ implies in 1bCS, one need to increase the number of measurements by a factor of $k$ - which leads to a much suboptimal result in recovery with false positives in 1bCS compared to group testing. In the case of exact recovery, this does not  affect much because of the nature of a measurement matrix and decoding algorithm~\cite{ABK17}; but that technique does not extend to recovery with false positives.

This leads us to believe that a binary measurement matrix may not be optimal in all settings of support recovery in 1bCS, although for support recovery using binary matrices is the standard~\cite{GNJN13,ABK17}. 
Indeed, using a carefully designed nonbinary matrix we can perform recovery with only small number of false positives using $O(k^{3/2} \log n)$ measurements. In this setting anything $o(k^2)$ was  elusive. On the other hand,
we show that using a binary matrix it is possible to do approximate recovery using $O(k\log n)$ measurements (a recovery that contains a small proportion of false positives and false negatives).  

\subsection{Notations} 

We write $[n]$ to denote the set $\{1,2,\dots,n\}$.
For any  $\fl{v} \in \bb{R}^n$, we use $\fl{v}_i$ to denote the $i^{\s{th}}$ coordinate of $\fl{v}$ and for any ordered set $\ca{S} \subseteq [n]$, we will use the notation $\fl{v}_{\mid \ca{S}} \in \bb{R}^{|\ca{S}|}$ to denote the vector $\fl{v}$ restricted to the indices in $\ca{S}$. Furthermore, we will use $\s{supp}(\fl{v}) \triangleq \{i \in [n]:\fl{v}_i \neq 0\}$ 
to denote the support of $\fl{v}$ and $\left|\left|\fl{v}\right|\right|_0 \triangleq \left|\s{supp}(\fl{v})\right|$ to denote the size of the support. We define the {\em dynamic range} $\kappa(\fl{v})$ of the vector $\fl{v}$ to be the ratio of the magnitudes of maximum and minimum non-zero entries of $\fl{v}$ i.e.
\begin{align*}
  \kappa(\fl{v}) \triangleq \frac{\max_{i \in [n]: \fl{v}_i \neq 0} \left|\fl{v}_i\right|}{\min_{i \in [n]: \fl{v}_i \neq 0} \left|\fl{v}_i\right|}.
\end{align*}
For a  vector  $\fl{v}\in \bb{R}^n$, let us denote by $\rho(\fl{v}) \triangleq \min(\left|\{i \in [n]:\fl{v}_i>0\}\right|,\left|\{i \in [n]:\fl{v}_i<0\}\right|)$,
the minimum number of non-zero entries of the same sign in $\fl{v}$. Finally, let $\s{sign}:\bb{R}\rightarrow \{-1,0,+1\}$ be a function that returns the sign of a real number i.e. for any input $x \in \bb{R}$, 
\begin{align*}
    \s{sign}(x) = 
\begin{cases}
 1 \quad \text{if $x > 0$}\\
 0 \quad \text{if $x = 0$}\\
 -1 \quad \text{if $x < 0$}
\end{cases}.
\end{align*}
Note that the range of the sign function has size $3,$ therefore using this at the output of a measurement will not technically be a $1$-bit information. Consider the true 1-bit sign function $\s{sign^\ast}:\bb{R}\rightarrow \{-1,+1\}$, where  
\begin{align*}
    \s{sign^\ast}(x) = 
\begin{cases}
 1 \quad \text{if $x \ge 0$}\\
 -1 \quad \text{if $x < 0$}
\end{cases}.
\end{align*}
It is possible to evaluate $\s{sign}(x)$ from $\s{sign^\ast}(x)$ and $\s{sign^\ast}(-x)$ for any $x \in \bb{R}$. Therefore all the results related to the $\s{sign}$ function holds for the $\s{sign^\ast}$ function with the number of measurements being within a factor of $2.$

Extending this notation for a vector $\fl{v}\in \bb{R}^n$, let $\s{sign}(\fl{v})\in \{-1,0,1\}^n$ be a vector comprising the signs of coordinates of $\fl{v}$. More formally, we have $\s{sign}(\fl{v})_i = \s{sign}(\fl{v}_i) $ for all $i \in [n]$. For any matrix $\fl{M}\in \bb{R}^{m \times n}$ 
and any set $\ca{S}\subseteq[n]$, we will write $\fl{M}_{\ca{S}}\in \bb{R}^{m \times \left|\ca{S}\right|}$ to denote the sub-matrix formed by the columns constrained to the indices in $\ca{S}$. We will write $\fl{M}_{ij}$ to denote the entry in the $i^{\s{th}}$ row and $j^{\s{th}}$ column of $\fl{M}$.
Finally, we will use $\s{col}(\fl{M})$ to denote the set of columns of the matrix $\fl{M}$. 

\subsection{Formal Problem Statement}\label{sec:prob_stat} 

 Consider an unknown sparse signal  $\fl{x} \in \bb{R}^n$ with $\left|\left|\fl{x}\right|\right|_0 \le k$. In the 1bCS framework, we design a sensing matrix $\fl{A}\in \bb{R}^{m \times n}$ to obtain the measurements of $\fl{x}$ as 
 \begin{align*}
     \fl{y} = \s{sign}(\fl{Ax}).
 \end{align*}

  In this work, we primarily consider the problem of \textit{support recovery} where our goal is to design the sensing matrix $\fl{A}$ with minimum number of measurements (rows of $\fl{A}$) so that we can recover the support of $\fl{x}$ from $\s{sign}(\fl{Ax})$. Our goal is to design \textit{universal} sensing matrices which fulfil a given objective \textit{for all} unknown $k$-sparse signal vectors. We look at three different notions of universal support recovery as defined below: 
  
  \begin{defn}[universal exact support recovery]
   A measurement matrix $\fl{A}\in \bb{R}^{m \times n}$ is called a \emph{universal exact support recovery} scheme if  there exists a recovery algorithm that,  for all $\fl{x}\in \bb{R}^n, \|\fl{x}\|_0 \le k$, returns $\s{supp}(\fl{x})$ on being provided $\s{sign}(\fl{Ax})$ as input. 
  \end{defn}
 

  \begin{defn}[universal $\epsilon$-approximate  support recovery]\label{def:approx}
    Fix any $0<\epsilon<1$. A measurement matrix $\fl{A}\in \bb{R}^{m \times n}$ is called a \emph{universal $\epsilon$-approximate support recovery} scheme if  there exists a recovery algorithm that, for all $\fl{x}\in \bb{R}^n, \|\fl{x}\|_0 \le k$, returns a set $\ca{S}\subseteq[n], \left|\ca{S}\right| \le k,$ satisfying $\left|\ca{S} \cap\s{supp}(\fl{x})\right| \ge \max(\left|\left|\fl{x}\right|\right|_0-\epsilon k,0)$ and $\left|\ca{S}\setminus\s{supp}(\fl{x})\right|\le \epsilon k$ on being provided $\s{sign}(\fl{Ax})$ as input. 
  \end{defn}
Evidently, the $\epsilon$-approximate  support recovery schemes allow for recovery with a small ($2\epsilon k$) number of errors (which may include $\epsilon k$ false positives and $\epsilon k$ false negatives).

  \begin{defn}[universal $\epsilon$-superset recovery]\label{def:superset}
   Fix any $0<\epsilon<1$. A measurement matrix $\fl{A}\in \bb{R}^{m \times n}$is called a \emph{universal $\epsilon$-superset recovery} scheme if  there exists a recovery algorithm that,  for all $\fl{x}\in \bb{R}^n, \|\fl{x}\|_0 \le k$, returns a set $\ca{S}\subseteq[n], \left|\ca{S}\right| \le \|\fl{x}\|_0+\epsilon k$ satisfying $\s{supp}(\fl{x}) \subseteq \ca{S}$ on being provided $\s{sign}(\fl{Ax})$ as input. 
  \end{defn}
  
  \begin{prop}\label{prop:compare}
  Any measurement matrix $\fl{A}\in \bb{R}^{m \times n}$ that is a universal $\epsilon$-superset recovery scheme is also a universal $\epsilon$-approximate recovery scheme.
  \end{prop}
  
  \begin{proof}
  Consider a measurement matrix $\fl{A}\in \bb{R}^{m \times n}$ that is a universal $\epsilon$-superset recovery scheme. This implies that there exists a recovery algorithm $\ca{A}$ that,  for all $\fl{x}\in \bb{R}^n, \|\fl{x}\|_0 \le k$, returns a set $\ca{S}\subseteq[n], \left|\ca{S}\right| \le \|\fl{x}\|_0+\epsilon k$ satisfying $\s{supp}(\fl{x}) \subseteq \ca{S}$ on being provided $\s{sign}(\fl{Ax})$ as input. Note that, when $\|\fl{x}\|_0+\epsilon k \le k$ this follows immediately. Now consider the case when $\|\fl{x}\|_0+\epsilon k > k$. For a fixed $\fl{x}\in \bb{R}^n, \|\fl{x}\|_0 \le k$, we can compute a set $\ca{S}'$ by deleting any $\epsilon k$ indices from the set $\ca{S}$ returned by Algorithm $\ca{A}$. Clearly, the set $\ca{S}'$ has a size of at most $\min(\left|\left|\fl{x}\right|\right|_0,k)$ and furthermore, $\left|\ca{S}'\cap \s{supp}(\fl{x})\right| \ge \left|\left|\fl{x}\right|\right|_0 - \epsilon k$ implying that $\left|\ca{S}'\setminus\s{supp}(\fl{x})\right| \le \epsilon k$. Hence $\fl{A}$ is a universal $\epsilon$-approximate recovery scheme.
  \end{proof}
  
  The $\epsilon$-superset recovery schemes allow for support recovery with only a small ($\epsilon k$) number of false positives and 0 false negative. As mentioned in \cite{flodin2019superset},  an $\epsilon$-superset recovery scheme makes subsequent approximate vector recovery computationally and statistically efficient, as instead of focusing on all $n$ coordinates, one can focus on only $O(k)$ coordinates.
  Furthermore, notice that Definition \ref{def:superset} poses a stricter recovery requirement than Definition \ref{def:approx}, and therefore should require more measurements.

We study measurement complexity of the three aforementioned notions of support recovery for  general $k$-sparse input signals, as well as for the setting where additional side information on the input vector $\fl{x}$ is known. 
In the later case, the following two scenarios were considered: 1) $\fl{x}$ has   dynamic range bounded by a known number 2) The minimum number of non-zero entries of $\fl{x}$ having the same sign is known to be bounded from above. The reason for considering these two scenarios is the following. The first generalizes the result for binary vectors (studied in \cite{ABK17}), and the second generalizes the result for positive vectors (studied in \cite{flodin2019superset}).

\subsection{Our Results}

Our main contribution is to provide algorithms and upper bounds on the measurement complexity for the three distinct notions of support recovery.
Our results (summarized in Table \ref{table:support}) resolve a number of open questions raised in \cite{flodin2019superset} and improves upon previously known bounds. 
Our main techniques involve utilizing novel modifications or generalization of well-known combinatorial structures such as Disjunct matrices and Cover-free families used primarily in group testing  literature~\cite{du2000combinatorial,ngo2011efficiently,barg2017group}. 

First, note that with $n$ measurements, it is always possible to recover the support trivially. For universal exact support recovery, the state of the art scheme with $O(k^2\log n)$ number of measurements is given by \cite{ABK17}. The construction is based on Robust Union-Free Families (RUFF), a set system with some combinatorial property that will be discussed later.
When it is known that the signal $x$ is binary (alternatively, a set of measurements that work for all binary vector $x \in \{0,1\}^n$), there exist a exact recovery scheme with $O(k^{3/2}\log(n/k))$ measurements~\cite{ABK17,JLBB13}. For this purpose, a set of Gaussian measurements are capable of universal recovery with high probability.

\paragraph{Universal $\epsilon$-superset recovery.} To reduce the number of measurements from the order of $k^2$ to $k$, recovering a superset is proposed in \cite{flodin2019superset}. However, the technique therein does not work for all signals, but only vectors with nonnegative coordinates. As pointed out in \cite{flodin2019superset}, universal $\epsilon$-superset recovery still takes $O(k^2 \log n)$ measurements.

Our main contribution is to use  
combinatorial designs to show a measurement complexity of $O(k^{3/2}\epsilon^{-1/2}\log (n/k))$ for universal $\epsilon$-superset recovery. 
We also prove that $\Omega\Big(\frac{k}{\epsilon} \Big(\log \frac{k}{\epsilon}\Big)^{-1}\log \frac{n}{\epsilon k}\Big)$ measurements are necessary for $\epsilon$-superset recovery.
This is a significant reduction in the gap between the upper bound and the linear lower bound; the dependence on $k$ is reduced to only $k^{3/2}$ in the upper bound.

When an upper bound on the the dynamic range is known, or the minimum non-zero entries of the unknown signal vector having same sign is known to be a constant, we improve the measurement complexity to $O(k\epsilon^{-1}\log (n/k))$. 

Note that, when we substitute $\epsilon = 1/k$ in the above two results, we see that for exact recovery we need $O(k^2 \log (n/k)$ measurements, recovering prior result. Therefore, our results give a smooth degradation in measurement complexity, as we seek a more accurate recovery.





\paragraph{Universal $\epsilon$-approximate support recovery.} For approximate recovery of support, no direct prior results exist, however any algorithm for  $\epsilon$-superset recovery provides $\epsilon$-approximate support recovery guarantee trivially. We introduce a generalization of the robust union free families,  namely List union-Free family and use its properties to show that  $O(k\epsilon^{-1}\log (n/k))$ measurements are sufficient in the general case, a strict improvement on the superset recovery. We also prove that this guarantee is tight up to logarithmic factors by showing that 
$\Omega\Big(\frac{k}{\epsilon} \Big(\log \frac{k}{\epsilon}\Big)^{-1}\log \frac{n}{\epsilon k}\Big)$ measurements are necessary for universal $\epsilon$-approximate recovery.

When the dynamic range of the unknown signal vector is  bounded from above by a known quantity $\eta$, we improve the measurement complexity to $O(k\eta\epsilon^{-1/2}\log (n\eta))$ (thus beating the lower bound above by a factor of $\frac{1}{\sqrt{\epsilon}}$). Note again that, if we substitute $\epsilon= 1/k$, we recover a generalization of the existing result on universal recovery for binary vectors, i.e., we recover the $k^{3/2}$ scaling.

Our results on sufficient number of measurements for universal support recovery are summarized in the table below.

\begin{table*}[htbp]
\centering
\begin{tabular}{ |c|c|c|c|c| c|} 
\hline
Problem &  $\fl{x} \in \bb{R}^n$ &  $\fl{x} \in \bb{R}^n: \kappa(\fl{x})\le \eta$ & $\fl{x} \in \{0,1\}^n$ & $\fl{x} \in \bb{R}^n$ (lower bound)\\
\hline
\hline
 Exact  & $O(k^2\log n)$ \cite{ABK17} & $O(k^2\log \frac{n}{k})$  & $O(k^{3/2}\log \frac{n}{k})$ \cite{ABK17} & $\Omega\Big(k^2 \frac{\log n}{\log k}\Big)$ \cite{ABK17} \\
 $\epsilon$-Approximate  & $O(\frac{k}{\epsilon}\log \frac{n}{k})$  & $O(\frac{k\eta}{\epsilon^{1/2}}\log (n\eta))$  & $O(\frac{k}{\epsilon^{1/2}}\log n)$ & $\Omega\Big(\frac{k}{\epsilon} \Big(\log \frac{k}{\epsilon}\Big)^{-1}\log \frac{n}{\epsilon k}\Big)$ \\
 $\epsilon$-Superset  & $O(\frac{k^{3/2}}{\epsilon^{1/2}}\log \frac{n}{k})$ & $O(\frac{k}{\epsilon}\log\frac{n}{k})$  & $O(\frac{k}{\epsilon}\log \frac{n}{k})$ & $\Omega\Big(\frac{k}{\epsilon} \Big(\log \frac{k}{\epsilon}\Big)^{-1}\log \frac{n}{\epsilon k}\Big)$ \\
\hline
\end{tabular}
\caption{\label{table:support}Our results for universal support recovery in $1$-bit Compressed Sensing for different settings and different class of signals. Rows 2 and 3 contain new results proved in this paper.}
\end{table*}

\subsection{Main Technical Contribution}

Our new technical contribution in the 1bCS support recovery problem 
is to use simple properties of polynomial roots in conjunction with combinatorial designs for designing measurements. More precisely, we design a row (say $\fl{z}$) of the measurement matrix $\fl{A}$ such that the non-zero entries of $\fl{z}$ are integral powers of some  number $\alpha\in \bb{R}$. The important insight that we now use in our algorithms is that the inner product of the unknown sparse signal and the measurement vector (i.e. $\langle \fl{x}, \fl{z} \rangle$) can be described as the evaluation of a polynomial whose coefficients are entries of $\fl{x}$ at the number $\alpha$. Recall that in Section \ref{sec:differences}, we argued that the main hurdle in the 1bCS setting (as compared to the group testing setting) is that it is difficult to interpret the meaning of a \texttt{0} output. From our construction of the measurement vector $\fl{z}$, the evaluation of a polynomial can be zero at  $\alpha$ if $\alpha$ is a root of the polynomial or the polynomial is everywhere \texttt{0}. 
Since the number of roots of a polynomial is finite, we can carefully design measurement vectors (with different $\alpha$'s) so that their inner product with $\fl{x}$ is the evaluation of the same polynomial but all of their output cannot be zero unless the polynomial is everywhere zero. This property allows us to precisely interpret what a \texttt{0} for all these group of measurements imply.  

We are left with bounding the number of roots of such polynomials. But the number of roots of a polynomial is at most the number of non-zero coefficients. The sparsity of $\fl{x}$ immediately implies that the number of roots of any such polynomial can be $k$. Let us start with a superset recovery (allows a few false positives) matrix for group testing (has $k\log n$ rows) and modify in the above way. However, designing $k$ measurements corresponding to each polynomial again lead to the $O(k^2 \log n)$ upper bound on the measurement complexity. In order to get around this issue, our second key idea is to do design a matrix for universal superset recovery  in two steps. First, we design a measurement matrix for universal approximate recovery (allows a few false positives and false negatives) by proposing a new combinatorial design (Definition \ref{defn:new})  that generalizes well studied measurement matrices in the literature and incorporates many useful properties. In the next step, we ignore the indices obtained in the first step and only seek to correct the false negatives. Since the number of false negatives is significantly smaller than $k$ (the total sparsity), the number of roots of the designed polynomials is also accordingly small. By carefully optimizing the number of measurements used in the two steps, we obtain the $k^{3/2}$ scaling for superset recovery.

It turns out that under other mild assumptions on the unknown sparse signal such as a known dynamic range $(\kappa(\fl{x})\le \eta)$ or a small number of non-zero entries of the same sign $(\rho(\fl{x}) \le \eta')$, we can also use other useful properties of the polynomial roots. In the former case, Cauchy's theorem says that the magnitude of the polynomial roots is bounded from below by $1+\eta$ while in the latter case, Descartes' rule of signs imply that the number of polynomial roots is bounded from above by $2\eta'$. In both cases, these properties allow us to prove nearly tight guarantees on the measurement complexity. Finally, because of the combinatorial structure and ease of manipulating polynomials,  our overall algorithm with such measurements is also efficient.

\paragraph{Organization.} The rest of the paper is organized as follows. In Sec.~\ref{design}, we define some set systems that will be used for constructing the universal measurement schemes. In particular, we show probabilistic existence of {\em list union-free} families. In Sec.~\ref{sec:main}, we provide our main results and detailed proof for approximate support recovery and superset recovery, in that order. Finally we conclude with a discussion on open problems in this area.

\section{Combinatorial Designs}\label{design}

In this section, we will start with a few definitions characterizing matrices with useful combinatorial properties.
 
\begin{defn}[List-disjunct matrix~\cite{dyachkov1983survey,ngo2011efficiently}] 
An $m \times n$ binary matrix $\fl{M}\in \{0,1\}^{m \times n}$ is a $(k,\ell)$-list disjunct matrix if for any two disjoint sets $S,T \subseteq \s{col}(M)$ such that $\left|S\right|=\ell,\left|T\right|=k$, there exists a row in $M$ in which some column from $S$ has a non-zero entry, but every column from $T$ has a zero.
\end{defn}
The following result characterizes the sufficient number of rows in list-disjunct matrices: 
\begin{lemma}[\cite{ngo2011efficiently}]\label{disjunct_exists}
 An $m \times n$ $(k,\ell)$-list disjunct matrix exists with 
 \begin{align*}
     m \le 2k\Big(\frac{k}{\ell}+1\Big)\Big(\log \frac{n}{k+\ell}+1\Big).
 \end{align*}
Moreover an $m \times n$ $(k,\ell)$-list disjunct matrix with $k\ge2\ell$  must satisfy, 
\begin{align*}
    m = \Omega\Big(\frac{k^2}{\ell} \Big(\log \frac{k^2}{\ell}\Big)^{-1}\log \frac{n-k}{\ell}\Big).
\end{align*}

\end{lemma}

Disjunct matrices ($\ell=1$) and list disjunct matrices have a rich history of being utilized in the group testing literature \cite{du2000combinatorial,dyachkov1983survey,ngo2011efficiently,PR11,Maz16}. The premise in group testing is very similar to 1-bit compressed sensing: $\fl{y} = \s{sign}(\fl{A}\fl{x})$ except that 
both $\fl{x} \in \{0,1\}^n$ and $\fl{A} \in \{0,1\}^{m \times n}$ are binary (note that, therefore, $\fl{y}\in \{0,1\}^n$ as well).

Consider the a measurement $y=\s{sign}(\langle\fl{a},\fl{x}\rangle).$ In group testing, $y=0$ implies $\s{supp}(\fl{a})\cap \s{supp}(\fl{x}) =\emptyset.$ However, in 1bCS, $y$ can be zero even when  $\s{supp}(\fl{a})\cap \s{supp}(\fl{x}) \ne \emptyset.$ This creates the main difficulty in importing tools of group testing being used in 1bCS. 

To tackle this, a set system called robust union-free family was proposed in \cite{ABK17}. We generalize that notion to propose a List union-free family.  

\begin{defn}[List union-free family, List union-free matrix]\label{defn:new} 
A family of sets $\ca{F} \equiv \{\ca{B}_1, \ca{B}_2,\dots, \ca{B}_n\}$ where each $\ca{B}_i \subset [m]$, $|\ca{B}_i|=d$ is an $(n,m,d,k,\ell,\alpha)$-list union-free family if for any pair of disjoint sets $S, T \subseteq [n]$ with $|S|= \ell,|T|=k$, there exists $j \in S$ such that $|\ca{B}_j \cap (\bigcup_{i \in (T\cup S)\setminus\{j\}} \ca{B}_i) | < \alpha |\ca{B}_j|$. 

Suppose, $\ca{F} \equiv \{\ca{B}_1, \ca{B}_2,\dots, \ca{B}_n\}$ is an $(n,m,d,k,\ell,\alpha)$-list union-free family. An $m \times n$ binary matrix $\fl{M}\in \{0,1\}^{m \times n}$ is a $(n,m,d,k,\ell,\alpha)$-list union-free matrix if the entry in the $i^{\s{th}}$ row and $j^{\s{th}}$ column of $\fl{M}$ is \texttt{1} if $i \in \ca{B}_j$ and \texttt{0} otherwise.
\end{defn}

Special cases of List union-free families, such as union-free families or cover-free codes ($(n,m,d,k,1,1)$-list union-free families) are well-studied~\cite{erdos1985families,d2002families,frankl1986union,coppersmith1998new,furedi1996onr}  
 and has found applications in cryptography and experiment designs. 
An $(n,m,d,k,1,\alpha)$-list union-free family is called a robust union-free family, and  it has been recently used for support recovery in 1bCS in \cite{ABK17}. The List union-free family that we introduce above is a natural generalization and  has not been studied previously to the best of our knowledge. We will show that this family of sets is useful for  universal superset recovery of support. Below, we provide a result that gives the sufficient number of rows in a List union-free matrix.
\begin{lemma}[Existence of list-union free matrices]\label{lem:list-union}
For a given $0<\alpha<1,n,k,\ell$, there exists a $(n,m,d,k,\ell,\alpha)$-list union-free matrix with number of rows
\begin{align*}
    &m = O\Big((k+\ell)\Big(\frac{e^2}{\alpha^3}\Big)\Big(\frac{k}{\ell}+1\Big)\Big(\log \frac{n}{k+\ell}+1\Big)\Big(\log \frac{e}{\alpha}\Big)^{-1}\Big)\Big) \\ 
    \text{and} \quad  
    &d = O\Big(\frac{1}{\alpha}\Big(\frac{k}{\ell}+1\Big)\Big(\log \frac{n}{k+\ell}+1\Big)\Big(\log \frac{e}{\alpha}\Big)^{-1}\Big)\Big).
\end{align*}

\end{lemma}

\begin{proof}
Let us fix $m'=m/q$ where $m,q$ is to be decided later. Consider an alphabet $\f{\Sigma}$ of size $q$ and subsequently, we construct a random matrix $\fl{M}' \in \f{\Sigma}^{m' \times n}$ where each entry is sampled independently and uniformly from $\f{\Sigma}$. 
We will write the $i^{\s{th}}$ column of the matrix $\fl{M}'$ in the form of a set of tuples $\ca{B}'_i \equiv \bigcup_{r \in [m']}\{(\fl{M}'_{ri},r)\}$. In other words, the symbol $\fl{M}'_{ri}$ in the $r^{\s{th}}$ row and $i^{\s{th}}$ column of $\fl{M}'$ is mapped to the tuple $(\fl{M}'_{ri},r)$ in $\ca{B}'_i$; hence $|\ca{B}'_i|=m'$ for all $i\in [n]$. Now, consider two disjoint sets $\ca{S},\ca{T} \subseteq \s{col}(\fl{M}')$ such that $|\ca{S}|=\ell, |\ca{T}|=k$. We will call $\ca{S},\ca{T}$ bad if 
\begin{align*}
\Big|\ca{B}'_i \bigcap \Big(\bigcup_{j \in (\ca{T}\cup \ca{S})\setminus \{i\}} \ca{B}'_j\Big)\Big| \ge \alpha m' \quad \text{for all }i\in \ca{S}.
\end{align*}
For a fixed $i \in \ca{S}$ and fixed $\ca{T}$, let us define the event $\ca{E}^{i,\ca{T}}\triangleq \{\left|\ca{B}'_i \cap (\cup_{(\ca{T}\cup \ca{S})\setminus \{i\}} \ca{B}'_j) \right| \ge \alpha m'$\}. Hence $\ca{S},\ca{T}$ (as defined above) is bad if $\bigcap_{i \in \ca{S}}\ca{E}^{i,\ca{T}}$ is true. Again, for a fixed $i \in \ca{S}$, consider any subset $\ca{S}'\subseteq \ca{S}\setminus\{i\}$.
We will have
\begin{align*}
    &\Pr(\ca{E}^{i,\ca{T}} \mid \bigcap_{i' \in \ca{S}'}\ca{E}^{i',\ca{T}}) \\ 
    &=\Pr\Big(\big|\ca{B}'_i \bigcap \Big(\bigcup_{j \in (\ca{T}\cup \ca{S})\setminus\{i\}} \ca{B}'_j\Big)\big| \ge \alpha m' \mid \bigcap_{i' \in \ca{S}'}\ca{E}^{i',\ca{T}}\Big) \\
    &\stackrel{(a)}{=}
\sum_{\ca{R}\in \Omega}\Pr\Big(\bigcup_{j \in (\ca{T}\cup \ca{S})\setminus\{i\}} \ca{B}'_j=\ca{R} \mid \bigcap_{i' \in \ca{S}'}\ca{E}^{i',\ca{T}}\Big)\Pr\Big(\Big|\ca{B}'_i \bigcap (\bigcup_{j \in (\ca{T}\cup \ca{S})\setminus\{i\}} \ca{B}'_j)\Big| \ge \alpha m'\mid \bigcup_{j \in (\ca{T}\cup \ca{S})\setminus\{i\}} \ca{B}'_j=\ca{R}, \bigcap_{i' \in \ca{S}'}\ca{E}^{i',\ca{T}}\Big)\\
&\stackrel{(b)}{\le} \sum_{\ca{R}\in \Omega}\Pr\Big(\bigcup_{j \in (\ca{T}\cup \ca{S})\setminus\{i\}} \ca{B}'_j=\ca{R} \mid \bigcap_{i' \in \ca{S}'}\ca{E}^{i',\ca{T}}\Big){m' \choose \alpha m'} \Big(\frac{k+\ell}{q}\Big)^{\alpha m'} \\
&\stackrel{(c)}{\le} {m' \choose \alpha m'} \Big(\frac{k+\ell}{q}\Big)^{\alpha m'}
\end{align*}
where the summation in steps (a) and (b) is over all elements in the sample space $\Omega$ of the random variable $\bigcup_{j \in (\ca{T}\cup \ca{S})\setminus\{i\}} \ca{B}'_j$. Step (a) 
follows from the law of total probability where we further condition on each value $\ca{R}$ of the random set $\bigcup_{j \in (\ca{T}\cup \ca{S})\setminus\{i\}} \ca{B}'_j$. Step (b) follows from the fact that for any value $\ca{R}$ of the random variable $\bigcup_{j \in (\ca{T}\cup \ca{S})\setminus\{i\}} \ca{B}'_j$, any row of the matrix $\fl{M}'$ restricted to the columns in $(\ca{T}\cup \ca{S})\setminus\{i\}$ can contain at most $k+\ell$ distinct symbols. Hence the probability that for a fixed row of $\fl{M}'$ , the symbol in $i^{\s{th}}$ column is contained in the set of symbols present in the columns in $(\ca{T}\cup \ca{S})\setminus\{i\}$ is at most $(k+\ell)/q$; therefore the probability that there exists at least $\alpha m'$ such rows is bounded from above by ${m' \choose \alpha m'} \Big(\frac{k+\ell}{q}\Big)^{\alpha m'}$. Step (c) follows from the fact that the sum of probabilities of all values of the random set $\bigcup_{j \in (\ca{T}\cup \ca{S})\setminus\{i\}} \ca{B}'_j$ conditioned on $\cap_{i' \in \ca{S}'}\ca{E}^{i',\ca{T}}$ is 1.

Let us denote the the distinct columns in $\ca{S}$ by $i_1,i_2,\dots,i_{\ell}$. Subsequently, we have  
\begin{align*}
\Pr(\text{$\ca{S},\ca{T}$ is bad})&=\Pr\Big(\bigcap_{t \in [\ell]} \ca{E}^{i_t,\ca{T}}\Big)
= \prod_{t \in [\ell]}\Pr\Big( \ca{E}^{i_t,\ca{T}}\mid \bigcap_{f \in [t-1]} \ca{E}^{i_f,\ca{T}}\Big) \le \Big({m' \choose \alpha m'}\Big(\frac{k+\ell}{q}\Big)^{\alpha m'}\Big)^{\ell}. 
\end{align*}

Hence, we get that

\begin{align*}
    &\Pr(\bigcup_{\ca{S},\ca{T}}\text{$\ca{S},\ca{T}$ is bad}) 
    \le \sum_{\ca{S},\ca{T}} \Pr(\text{$\ca{S},\ca{T}$ is bad}) \\
    &\le {n \choose k+\ell}{k+\ell \choose \ell} \Big({m' \choose \alpha m'} \Big(\frac{k+\ell}{q}\Big)^{\alpha m'}\Big)^{\ell} \\
    &\le \exp\Big((k+\ell)\log \frac{en}{k+\ell}+\ell\log \frac{e(k+\ell)}{\ell}+\ell m' \alpha\log \frac{e}{\alpha}-\alpha m' \ell \log \frac{q}{k+\ell}\Big).
\end{align*}

Now, we choose
\begin{align*}
 q=\Big\lceil (k+\ell)\Big(\frac{e}{\alpha}\Big)^{2} \Big\rceil   \quad  \text{and} \quad m' = \frac{2}{\alpha}\Big(\frac{k}{\ell}+1\Big)\Big(\log \frac{n}{k+\ell}+e\Big) \Big(\log \frac{e}{\alpha}\Big)^{-1} 
\end{align*}
in which case we get that $\Pr(\bigcup_{\ca{S},\ca{T}}\text{$\ca{S},\ca{T}$ is bad}) <1$. This implies that there exists a matrix $\fl{M}'$ with $m'$ rows such that no pair of disjoint sets $\ca{S},\ca{T}$ with $|\ca{S}|=\ell,|\ca{T}|=k$ is bad.
Let us denote the standard basis vectors in $\bb{R}^q$ by $\fl{e}^1,\fl{e}^2,\dots,\fl{e}^q$; $\fl{e}^i$ represents the $q$-dimensional vector such that the $i^{\s{th}}$ entry is $1$ and all other entries are $0$. Consider any fixed ordering of the symbols in $\f{\Sigma}$; for the $i^{\s{th}}$ symbol in $\f{\Sigma}$, we will map it to the vector $\fl{e}^{i}$. We can now construct the matrix $\fl{M}\in \{0,1\}^{m \times n}$ from $\fl{M}'$ by replacing each symbol in $\f{\Sigma}$ with the corresponding vector in the standard basis of $\bb{R}^q$ based on the aforementioned mapping. Clearly, each column in this matrix has $d=m'$ \texttt{1}'s. Moreover, for any $i\in [m]$ and $j,v \in [n]$, we will have $\fl{M}_{ij}=\fl{M}_{iv}=1$ if and only if $\fl{M}'_{i'j}=\fl{M}'_{i'v}=s$ where $i'= \lceil i/q \rceil$ and $s$ is the $(i \pmod{q})^{\s{th}}$ symbol in $\f{\Sigma}$.
Let us denote by $\ca{B}_i \subseteq [m]$ the indices of the rows where $i^{\s{th}}$ column of $\fl{M}$ has non-zero entries. In that case, $\left|\ca{B}_i\right|=m'$ for all $i\in [n]$ and furthermore, for any pair of disjoint sets $S, T \subseteq [n]$ with $|S|= \ell,|T|=k$, there exists $j \in S$ such that $|\ca{B}_j \cap (\bigcup_{i \in {(T\cup S)\setminus\{i\}}} \ca{B}_i) | < \alpha |\ca{B}_j|$. 
Hence, the matrix $\fl{M}$ is also a $(n,m,d,k,\ell,\alpha)$-list union-Free matrix with 
\begin{align*}
    &m = O\Big((k+\ell)\Big(\frac{e^2}{\alpha^3}\Big)\Big(\frac{k}{\ell}+1\Big)\Big(\log \frac{n}{k+\ell}+1\Big)\Big(\log \frac{e}{\alpha}\Big)^{-1}\Big)\Big) \\ 
    \text{and} \quad  
    &d = O\Big(\frac1\alpha\Big(\frac{k}{\ell}+1\Big)\Big(\log \frac{n}{k+\ell}+1\Big)\Big(\log \frac{e}{\alpha}\Big)^{-1}\Big)\Big).
\end{align*}

\end{proof}




\section{Recovery Algorithms and Results}\label{sec:main}
We first describe our results and techniques for approximate support recovery, followed by superset recovery; because the first uses a simpler algorithm than the later, supposedly harder problem.
\subsection{Approximate Support Recovery}

The following is a result on universal $\epsilon$-approximate  support recovery for all unknown $k$-sparse signal vectors $\fl{x}\in \bb{R}^n$. The relevant recovery algorithm is given in Algorithm~\ref{algo:supp_approx}.
\begin{thm}\label{thm:first}
There exists a $1$-bit compressed sensing matrix $\fl{A}\in \bb{R}^{m \times n}$ for universal $\epsilon$-approximate support recovery of all $k$-sparse signal vectors with $m=O(k\epsilon^{-1} \log (n/k))$ measurements. Moreover the support recovery algorithm (Algorithm \ref{algo:supp_approx}) has a running time of $O(n\epsilon^{-1} \log (n/k))$. 
\end{thm}  
 
\begin{algorithm}[htbp]
\caption{\textsc{Approximate Support Recovery    }($\epsilon$)   \label{algo:supp_approx}}
\begin{algorithmic}[1]
\REQUIRE $\fl{y}=\s{sign}(\fl{Ax})$ 
\STATE Set $\ca{C}=\phi$.
\FOR{$j\in [n]$}
\IF{$\left|\ca{B}_j\cap \s{supp}(\fl{y})\right| \ge d/2$}
\STATE  $\ca{C} \leftarrow \ca{C}\cup \{j\}$
\ENDIF
\ENDFOR
\STATE Compute and return $\ca{C}'$ by deleting any $\max(0,\left|\ca{C}\right|-k)$ indices from $\ca{C}$.
\end{algorithmic}
\end{algorithm}

\begin{proof}

 Let $\fl{A}$ be a $(n,m,d,k,\epsilon k/2,0.5)$-list union-free matrix constructed from a $(n,m,d,k,\epsilon k/2,0.5)$-list union-free family $\ca{F}= \{\ca{B}_1, \ca{B}_2,\dots, \ca{B}_n\}$. From Lemma \ref{lem:list-union} (by substituting $\ell=\epsilon k/2, \alpha = 0.5$), we know that such a matrix $\fl{A}$ exists with $d=O(\epsilon^{-1} \log (n/k))$ and $m=O(k\epsilon^{-1} \log (n/k))$ rows. 

For the rest of the proof, we will simply go over the correctness of the recovery process, i.e., Algorithm~\ref{algo:supp_approx}. 
\paragraph{Correctness of recovery algorithm.} Fix a particular unknown signal vector $\fl{x}\in \bb{R}^n$. Recall that we obtain the measurements $\fl{y}=\s{sign}(\fl{A}\fl{x})$. Consider any set of indices $\ca{S} \subseteq [n]$ such that $|\ca{S}|=\epsilon k/2$ and $\ca{S} \cap  \s{supp}(\fl{x}) = \emptyset.$  Using the properties of the family $\ca{F}$, there exists an index $j \in \ca{S}$ such that 
\begin{align*}
    &\left|\ca{B}_j \setminus \Big(\bigcup_{i \in (\s{supp}(\fl{x})\cup S)\setminus \{j\}} \ca{B}_i\Big)\right| = \left|\ca{B}_j \right|- \left|\ca{B}_j \cap \Big(\bigcup_{i \in (\s{supp}(\fl{x})\cup S)\setminus \{j\}} \ca{B}_i\Big)\right|\ge \frac{d}{2} \\
    &\implies \left|\ca{B}_j \setminus \Big(\bigcup_{i \in \s{supp}(\fl{x})} \ca{B}_i\Big)\right| \ge \frac{d}{2}
\end{align*}
This implies that there exists at least $d/2$ rows in $\fl{A}$ where the $j$th entry is \texttt{1} but all the entries belonging to the support of $\fl{x}$ is \texttt{0}. For all these rows used as measurements, the output must be \texttt{0}. 
Using the fact that $\left|\ca{B}_j\right|=d$, we must have $\left|\s{supp}(\fl{y})\cap \ca{B}_j\right| < d/2$. On the other hand, consider a set  of indices $\ca{S} \subseteq \s{supp}(\fl{x})$ such that $|\ca{S}|=\epsilon k/2$. By using the property of the family $\ca{F}$, with $\ca{T}=\s{supp}(\fl{x})\setminus \ca{S}$, there must exist $j \in \ca{S}$ such that 
\begin{align*}
    \left|\ca{B}_j \setminus \Big(\bigcup_{i \in (\ca{T}\cup \ca{S})\setminus\{j\} } \ca{B}_i\Big)\right|=\left|\ca{B}_j \setminus \Big(\bigcup_{i \in \s{supp}(\fl{x})\setminus\{j\} } \ca{B}_i\Big)\right| = \left|\ca{B}_j \right|- \left|\ca{B}_j \bigcap \Big(\bigcup_{i \in \s{supp}(\fl{x})\setminus \{j\}} \ca{B}_i\Big)\right|\ge \frac{d}{2}.
\end{align*}
Therefore there exists 
 at-least $d/2$ rows where the $j$th entry is \texttt{1} but all the entries belonging to $\s{supp}(\fl{x})\setminus \{j\}$  is \texttt{0}; for all these rows used as measurements, the output must be \texttt{1}. Again, using the fact that $\left|\ca{B}_j\right|=d$, we must have $\left|\s{supp}(\fl{y})\cap \ca{B}_j\right| \ge d/2$.
 Therefore, if we compute $\ca{C}=\{j \in [n]:\left|\s{supp}(\fl{y})\cap \ca{B}_j\right| \ge d/2\}$, then $\ca{C}$ must satisfy the following properties:
1) $\left|\ca{C}\right|\le \|\fl{x}\|_0+k\epsilon/2 \le k(1+\epsilon/2)$, 2) $\left|\ca{C}\cap \s{supp}(\fl{x})\right| \ge \|\fl{x}\|_0-k\epsilon/2$ implying that $\ca{C}$ has large intersection with the support of $\fl{x}$, 3) $\left|\ca{C}\setminus \s{supp}(\fl{x})\right| \le \epsilon k/2$ implying that $\ca{C}$ has very few indices outside the support of $\fl{x}$. If $\left|\ca{C}\right| \le k$, then $\ca{C}$ already satisfies the conditions for $\epsilon$-approximate support recovery. Now suppose that $\left|\ca{C}\right| > k$. In that case, if we delete any arbitrary $\max(0,\left|\ca{C}\right|-k)$ indices from $\ca{C}$ to construct $\ca{C}'$ such that $\left|\ca{C}'\right|\le k$, then $$\left|\ca{C}'\cap \s{supp}(\fl{x})\right| = \left|\ca{C}\cap \s{supp}(\fl{x})\right|- \max(0,\left|\ca{C}\right|-k)\ge \|\fl{x}\|_0-\epsilon k$$ implying that $\left|\ca{C}'\setminus\s{supp}(\fl{x}) \right|\le \epsilon k$. Finally, note that for each $j \in [n]$, it takes $O(d)=O(\epsilon^{-1} \log (n/k))$ time to compute $\left|\ca{B}_j\cap \s{supp}(\fl{y})\right| $. 
Therefore the time complexity of Algorithm \ref{algo:supp_approx} is $O(n\epsilon^{-1} \log (n/k))$.
\end{proof}
Next, we show an improvement in the sufficient number of measurements if an upper bound on the dynamic range of $\fl{x}$ is known apriori. 

\begin{thm}\label{thm:second}
There exists a $1$-bit compressed sensing matrix $\fl{A}\in \bb{R}^{m \times n}$ for $\epsilon$-approximate universal support recovery of all $k$-sparse unit norm signal vectors $\fl{x}\in \bb{R}^n$ such that $\kappa(\fl{x})\le \eta$ for some known $\eta>1,$ with $m=O(k\eta\epsilon^{-1/2} \log (n\eta))$ measurements. 
\end{thm}

The proof of Theorem \ref{thm:second} follows from using random Gaussian measurements and has been delegated to Appendix \ref{sec:appendix}. Note that,  for exact support recovery Theorem \ref{thm:second} gives the number of measurements to be $O(k^{3/2}\eta \log (n\eta))$, a generalization of the binary input result.

\subsection{Superset Recovery}

In this subsection we prove our main result on universal $\epsilon$-superset recovery for all unknown $k$-sparse signal vectors $\fl{x}\in \bb{R}^n$. For simplicity of exposition, for any set $\ca{X} \subseteq [n]$ and for any fixed unknown signal $\fl{x}$, we will call any index that lies in $\ca{X}\setminus \s{supp}(\fl{x})$ to be a \textit{false positive} of $\ca{X}$ and any index that lies in $\s{supp}(\fl{x})\setminus \ca{X}$ to be a \textit{false negative} of $\ca{X}$.

\begin{algorithm}[htbp]
\caption{\textsc{Superset Recovery}($\epsilon$) \label{algo:supp_superset}}
\begin{algorithmic}[1]
\REQUIRE $\fl{y}^1=\s{sign}(\fl{A}^{(1)}\fl{ x}),\fl{y}^2=\s{sign}(\fl{A}^{(2)} \fl{x})$ where $\fl{A}^{(1)},\fl{A}^{(2)}$ is constructed as described in proof of Theorem \ref{thm:third}.
\STATE Set $\ca{C}=\phi,\ca{C}'=[n]$.
\FOR{$j \in [n]$}
\IF{$\left|\ca{B}_j\cap \s{supp}(\fl{y^1})\right|>d/2$}
\STATE  $\ca{C} \leftarrow \ca{C}\cup \{j\}$
\ENDIF
\ENDFOR
\STATE Update $\ca{C}$ by deleting any $\max(0,\left|\ca{C}\right|-k)$ indices from $\ca{C}$.
\FOR{\text{ each row $\fl{z}$ in $\fl{B}$}}
\IF{$\s{supp}(\fl{z})\cap \ca{C}=\phi$ and $\s{sign}(\langle \fl{z}^i, \fl{x} \rangle) = 0 \quad \text{for all } i\in [\zeta k]$}
\STATE $\ca{C}' \leftarrow \ca{C}'\setminus \s{supp}(\fl{z})$
\ENDIF
\ENDFOR
\STATE Return $\ca{C}' \cup \ca{C}$.
\end{algorithmic}
\end{algorithm}

\begin{thm}\label{thm:third}
There exists a $1$-bit compressed sensing matrix $\fl{A}\in \bb{R}^{m \times n}$ for  universal $\epsilon$-superset recovery of all $k$-sparse signal vectors with $m=O(k^{3/2}\epsilon^{-1/2} \log (n/k))$ measurements. Moreover the recovery algorithm (Algorithm \ref{algo:supp_superset}) has a running time of $O(nk^{3/2}\epsilon^{-1/2}\log (n/k))$.
\end{thm} 

\begin{proof}
Let $0<\zeta <1$ be some number that will be determined later. The sensing matrix $\fl{A}$ is designed to be two matrices $\fl{A}^{(1)}$ and $\fl{A}^{(2)}$ (with distinct properties) stacked vertically i.e.
$$
\fl{A} =\begin{pmatrix}
\fl{A}^{(1)}\\
\fl{A}^{(2)}
\end{pmatrix}.
$$
The matrix $\fl{A}^{(1)}\in \{0,1\}^{v \times n}$ is designed to be a $(n,v,d,k,\zeta k/2,0.5)$-list union-free matrix constructed from a $(n,v,d,k,\zeta k/2,0.5)$-list union-free family $\ca{F}= \{\ca{B}_1, \ca{B}_2,\dots, \ca{B}_n\}$. From Lemma \ref{lem:list-union}, $\fl{A}^{(1)}$ is known to exist with $d=O(\zeta^{-1} \log (n/k))$ and at most $v=O(k\zeta^{-1} \log (n/k))$ rows. $\fl{A}^{(2)}$ is designed in the following manner: consider a binary $(k(1+\zeta),0.5\epsilon k)$-list disjunct matrix $\fl{B}\in \{0,1\}^{u \times n}$ which is known to exist with  $u=O(k\epsilon^{-1} \log (n/k))$ rows (by using Lemma \ref{disjunct_exists}).  For each row $\fl{z}\in \{0,1\}^n$ of $\fl{B}$, we choose $\zeta k$ distinct positive numbers $a_1,a_2,\dots,a_{\zeta k} \in \bb{R}_{+}$ and subsequently, we construct $\zeta k$ rows of $\fl{A}^{(2)}$ denoted by $\fl{z}^1 ,\fl{z}^2,\dots,\fl{z}^{\zeta k} \in \bb{R}^n$ as follows: for all $i \in [\zeta k], j \in [n]$, we have that
\begin{align*}
    &\fl{z}^i_j = 0 \quad \text{if } \fl{z}_j = 0 \\
    &\fl{z}^i_j = a_i^{t-1} \quad \text{if } \text{$j^{\s{th}}$ entry of $\fl{z}$ is the $t^{\s{th}}$ \texttt{1} in $\fl{z}$ from left to right}.
\end{align*}
In essence each row of $\fl{B}$ is mapped to $\zeta k$ rows of $\fl{A}^{(2)}$. Hence, the total number of rows in $\fl{A}^{(2)}$ is $O(k^2\zeta\epsilon^{-1} \log (n/k))$ and thus, the number of rows in $\fl{A}$ is $O((k^2\zeta\epsilon^{-1}+k\zeta^{-1}) \log (n/k))$. Setting $\zeta=\sqrt{\epsilon/k}$, we obtain the number of measurements to be $m=O(k^{3/2}\epsilon^{-1/2} \log (n/k))$.

In the remainder of the proof, we show the correctness of Algorithm~\ref{algo:supp_superset} when used along with the measurement matrix as constructed.
\paragraph{Correctness of the recovery algorithm.} Suppose  $\fl{x}\in \bb{R}^n$ is the  unknown signal vector. Using Theorem \ref{thm:first} (Algorithm \ref{algo:supp_approx}), we can compute a set $\ca{C}, \left|\ca{C}\right| \le k$ from $\s{sign}(\fl{A}^{(1)} \fl{x})$ such that 
$\left|\ca{C}\cap \s{supp}(\fl{x})\right| \ge \|\fl{x}\|_0-k\zeta$ and $\left|\ca{C}\setminus \s{supp}(\fl{x})\right| \le \zeta k$ implying that the set $\ca{C}$ has at most $\zeta k$ false negatives and $\zeta k$ false positives. In the subsequent steps of our decoding algorithm, our objective is to correct the aforementioned false negatives. To do so, we ignore all the measurements (rows of $\fl{A}^{(2)}$) whose support has a non-empty intersection with the set $\ca{C}$ computed in the first stage. Consider any set of indices $\ca{S} \subseteq [n], \ca{S}\cap (\ca{C}\cup \s{supp}(\fl{x}))=\phi$ such that $|\ca{S}|=\epsilon k/2$. By using the property of $(k(1+\zeta),\epsilon k/2)$-list disjunct matrix and the fact that $|\ca{C}\cup \s{supp}(\fl{x})| \le k(1+\zeta)$, there exists an index $j \in \ca{S}$ and a row $\fl{z}$ in $\fl{B}$ such that the support of $\fl{z}$ is disjoint from $\ca{C}\cup \s{supp}(\fl{x})$ and contains $j \in \ca{S}$. Therefore, there must exist $\zeta k$ corresponding rows in $\fl{A}^{(2)}$ (recall the construction of $\fl{A}^{(2)}$) denoted by $\fl{z}^1,\fl{z}^2,\dots,\fl{z}^{\zeta k}$ such that the support of these rows are disjoint from $\ca{C}\cup \s{supp}(\fl{x})$ and contains $j \in \ca{S}$. If
\begin{align*}
\s{sign}(\langle \fl{z}^i, \fl{x} \rangle) = 0 \quad \text{for all } i\in [\zeta k]. 
\end{align*}
then we will infer the entire support of $\fl{z}$ to be disjoint from the support of $\fl{x}$ and delete those indices from $\ca{C}'$ (see Algorithm \ref{algo:supp_superset}). Therefore, for any set $\ca{S}\subseteq [n]: |\ca{S}|=\epsilon k/2, \ca{S}\cap (\ca{C}\cup \s{supp}(\fl{x}))=\phi$, we can identify correctly at least one index $j \in \ca{S}$ that lies outside $\ca{C}\cup \s{supp}(\fl{x})$; subsequently, we will delete this index. On the other hand, as we will show, we will never delete any index that lies in $\s{supp}(\fl{x})\setminus \ca{C}$. At the end, we return the surviving indices plus the set $\ca{C}$ that we recovered in the first stage of decoding.

Consider any row $\fl{z}\in \fl{B}$ such that $\s{supp}(\fl{z})\cap \ca{C}=\phi$ and the corresponding measurements $\fl{z}^1,\fl{z}^2,\dots,\fl{z}^{\zeta k}$ in $\fl{A}^{(2)}$ (parameterized by $a_1,a_2,\dots,a_{\zeta k}$ respectively and have the same support as that of $\fl{z}$). Notice that for all $i \in [\zeta k]$ 
\begin{align*}
    \langle \fl{x}, \fl{z}^i\rangle = \sum_{t \in (\s{supp}(\fl{x})\setminus \ca{C}) \cap \s{supp}(\fl{z})} \fl{x}_t \fl{z}^i_t.  
\end{align*}
For all $i\in [\zeta k]$, the entries of the vector $\fl{z}^i$ are integral powers of some number $a_i$ and $|(\s{supp}(\fl{x})\setminus \ca{C}) \cap \s{supp}(\fl{z}^i)| \le |\s{supp}(\fl{x})\setminus \ca{C}|\le \zeta k$. Therefore,  $ \langle \fl{x}, \fl{z}^i\rangle$ is the evaluation of a polynomial 
\begin{align*}
    p(r) = \sum_{t \in (\s{supp}(\fl{x})\setminus \ca{C}) \cap \s{supp}(\fl{z})} \fl{x}_t r^{\alpha_t}.
\end{align*}
of degree at most $n-1$ and having at most $\zeta k$ non-zero coefficients at the number $a_i$ i.e. $ \langle \fl{x}, \fl{z}^i\rangle = p(a_i)$. Clearly, if $|(\s{supp}(\fl{x})\setminus \ca{C}) \cap \s{supp}(\fl{z})|=0$ then $\langle \fl{x}, \fl{z}^i\rangle =0$ for all $i \in [\zeta k]$. On the other hand, if $|(\s{supp}(\fl{x})\setminus \ca{C}) \cap \s{supp}(\fl{z})| \neq 0$, then $\langle \fl{x}, \fl{z}^i\rangle \neq 0$ for all $i \in [\zeta k]$. This is because the polynomial $p(r)$ with at most $\zeta k$ non-zero coefficients can have at most $\zeta k-1$ positive roots (using Descartes' rule of signs) which means that not all of $a_1,a_2,\dots,a_{\zeta k}$ (distinct positive numbers) can be roots of $p(r)$. Therefore the surviving indices in the second stage must consist of the false negatives from the first stage and at most $\epsilon k/2$ false positives. Hence, the final set $\ca{C}'\cup\ca{C}$ returned by Algorithm \ref{algo:supp_superset} will not contain any false negatives but can contain at most $\epsilon k/2 +\zeta k$ false positives. Since $\zeta$ was chosen to be $\sqrt{\epsilon/k}$,  the total number of false positives is at most $\epsilon k$.

Finally, note that Lines 3-5 in Algorithm \ref{algo:supp_superset} has a time complexity of $O((k/\epsilon)^{1/2} \log (n/k))$ and therefore Lines $2-5$ has a time complexity of $O(n(k/\epsilon)^{1/2} \log (n/k))$. Line 8 has a time complexity of $O(n(k\epsilon)^{1/2})$ and therefore Lines 7-11 has a time complexity of $O(nk^{3/2}\epsilon^{-1/2}\log (n/k))$ which dominates the time complexity of the algorithm. This completes the proof of the theorem.
\end{proof}

It turns out that if additional weak assumptions hold true for the unknown signal vector $\fl{x}\in \bb{R}^n$, then we can improve the sufficient number of measurements significantly. More formally, we have the following two theorems.

\begin{thm}\label{thm:fourth}
There exists a $1$-bit compressed sensing matrix $\fl{A}\in \bb{R}^{m \times n}$ for universal $\epsilon$-superset recovery of all $k$-sparse signal vectors $\fl{x}\in \bb{R}^n$ such that $\kappa(\fl{x})\le \eta$ for some known $\eta>1$, with $m=O(k\epsilon^{-1} \log (n/k))$ measurements. Moreover the recovery algorithm (Algorithm \ref{algo:supp_superset2}) has a running time of $O(nk\epsilon^{-1} \log (n/k))$. 
\end{thm}

\begin{proof}

The sensing matrix $\fl{A}$ in designed 
similarly to $\fl{A}^{(2)}$ as described in the proof of Theorem \ref{thm:third}. Consider a binary $(k,\epsilon k)$-list disjunct matrix $\fl{B}\in \{0,1\}^{m \times n}$ which is known to exist with at most $m=O(k\epsilon^{-1} \log (n/k))$ rows (see Lemma \ref{disjunct_exists}).  For each row $\fl{z}\in \{0,1\}^n$ of $\fl{B}$, we choose a positive number $a_{\fl{z}} > 1+\eta$ and subsequently, we construct a row of $\fl{A}$ denoted by $\fl{z}'$ as follows: for all $j \in [n]$, we have
\begin{align*}
    &\fl{z}'_j = 0 \quad \text{if } \fl{z}_j = 0 \\
    &\fl{z}'_j = a_{\fl{z}}^{t-1} \quad \text{if } \text{$j^{\s{th}}$ entry of $\fl{z}$ is the $t^{\s{th}}$ \texttt{1} in $\fl{z}$ from left to right}.
\end{align*}
In essence, each row of $\fl{B}$ is mapped to a unique row of $\fl{A}$; hence the total number of rows in $\fl{A}$ is also $O(k\epsilon^{-1} \log (n/k))$.
\begin{algorithm}[htbp]
\caption{\textsc{Superset Recovery with Bounded Dynamic Range}($\epsilon$) \label{algo:supp_superset2}}
\begin{algorithmic}[1]
\REQUIRE $\fl{y}=\s{sign}(\fl{A x})$ where $\fl{A}$ is constructed as described in proof of Theorem \ref{thm:fourth}.
\STATE Set $\ca{C}=[n]$.
\FOR{\text{ each row $\fl{z}$ in $\fl{B}$}}
\IF{$\s{sign}(\langle \fl{z}', \fl{x} \rangle) = 0$}
\STATE $\ca{C} \leftarrow \ca{C}\setminus \s{supp}(\fl{z})$
\ENDIF
\ENDFOR
\STATE Return $\ca{C}$.
\end{algorithmic}
\end{algorithm}
The rest of the proof will show the correctness of Algorithm~\ref{algo:supp_superset2}.
\paragraph{Correctness of the recovery algorithm.}
The inner product of any row $\fl{z}$ of $\fl{A}$ and the unknown signal vector $\fl{x}$ can be represented as the evaluation of a polynomial $p(r)$ (whose coefficients are the entries of $\fl{x}$) at $a_{\fl{z}}$ i.e.
\begin{align*}
    p(r) = \sum_{t \in \s{supp}(\fl{x}) \cap \s{supp}(\fl{z})} \fl{x}_t r^{\alpha_t}
\end{align*}
and $\langle \fl{z},\fl{x} \rangle = p(a_{\fl{z}})$. By using Cauchy's Theorem, we know that the magnitude of the roots of this polynomial $p(r)$ must be bounded from above by $1+\kappa(\fl{x}) \le 1+\eta$; hence $a_{\fl{z}}>1+\eta$ can never be a root of $p(r)$ unless it is always zero. Hence, the evaluation of $p(r)$ at $a_{\fl{z}}$ can be zero if and only if the polynomial $p(r)$ is always zero implying that the support of $\fl{z}$ must be disjoint from the support of $\fl{x}$. In other words, in Algorithm \ref{algo:supp_superset2}, we will never delete any indices that belong to the support of $\fl{x}$.

On the other hand, consider any set of indices $\ca{S} \subseteq [n]$ such that $|\ca{S}|=\epsilon k$ and $\ca{S}\cap \s{supp}(\fl{x})=\phi$. By using the property of $(k,\epsilon k)$-list disjunct matrix $\fl{B}$, there exists an index $j \in \ca{S}$ and a row $\fl{z}$ in $\fl{B}$ such that the support of $\fl{z}$ is disjoint from $\s{supp}(\fl{x})$ and contains $j \in \ca{S}$. Therefore, we will delete all indices in the support of $\fl{z}$ including the index $j$ from the set $\ca{C}$ in Step 4 of Algorithm \ref{algo:supp_superset2}.
Hence, we will delete all indices that belongs to the set $[n]\setminus \s{supp}(\fl{x})$ except at-most $\epsilon k$ of them. Therefore, the set $\ca{C}$ of surviving indices at the end of Algorithm \ref{algo:supp_superset2} satisfies the conditions for $\epsilon$-superset recovery.

Finally note that Line 4 in Algorithm \ref{algo:supp_superset2} has a time complexity of $O(n)$ and therefore, the total time complexity of the algorithm is $O(nk\epsilon^{-1} \log (n/k))$. This completes the proof of the theorem. 
\end{proof}

Finally, we give a result that concerns $\rho(\fl{x})$, the  minimum number of non-zero entries of the same sign in $\fl{x}$. This shows a generalization of the result known for only fully positive vectors.
\begin{thm}\label{thm:fifth}
There exists a $1$-bit compressed sensing matrix $\fl{A}\in \bb{R}^{m \times n}$ for universal $\epsilon$-superset recovery of all $k$-sparse signal vectors $\fl{x}\in \bb{R}^n$ such that $\rho(\fl{x})\le R$ for some known $R$, with $m=O(k\max(1,R)\epsilon^{-1} \log (n/k))$ measurements. Moreover the decoding algorithm (Algorithm \ref{algo:supp_superset3}) has a running time of $O(nk\max(1,R)\epsilon^{-1} \log (n/k))$. 
\end{thm}


\begin{proof}
As before, we will denote our sensing matrix by $\fl{A}$. Consider a binary $(k,\epsilon k)$-list disjunct matrix $\fl{B}\in \{0,1\}^{m \times n}$ which is known to exist with at-most $m=O(k\epsilon^{-1} \log (n/k))$ rows (see Lemma \ref{disjunct_exists}).  For each row $\fl{z}\in \{0,1\}^n$ of $\fl{B}$, we choose $2R+1$ distinct positive numbers $a_1, a_2, \dots, a_{2R+1}>0$. 
 Subsequently, we construct $R'=2R+1$ rows of $\fl{A}$ denoted by $\fl{z}^1,\fl{z}^2,\dots,\fl{z}^{R'}$ as follows: for all $i \in [R']$, $j \in [n]$, we have
\begin{align*}
    &\fl{z}^i_j = 0 \quad \text{if } \fl{z}_j = 0 \\
    &\fl{z}^i_j = a_i^{t-1} \quad \text{if } \text{$j^{\s{th}}$ entry of $\fl{z}$ is the $t^{\s{th}}$ \texttt{1} in $\fl{z}$ from left to right}.
\end{align*}
Hence, each row of $\fl{B}$ is mapped to $R'$ rows of $\fl{A}$ and therefore the total number of measurements is at most $O(kR'\epsilon^{-1}\log nk^{-1})$.
\begin{algorithm}[htbp]
\caption{\textsc{Superset Recovery with small minimum same sign entries }($\epsilon$) \label{algo:supp_superset3}}
\begin{algorithmic}[1]
\REQUIRE $\fl{y}=\s{sign}(\fl{A x})$ where $\fl{A}$ is constructed as described in proof of Theorem \ref{thm:fifth}.
\STATE Set $\ca{C}=[n], R' = \max(1,2R)$.
\FOR{\text{ each row $\fl{z}$ in $\fl{B}$}}
\IF{$\s{sign}(\langle \fl{z}^i, \fl{x} \rangle) = 0 \text{ for all } i\in [R']$}
\STATE $\ca{C} \leftarrow \ca{C}\setminus \s{supp}(\fl{z})$
\ENDIF
\ENDFOR
\STATE Return $\ca{C}$.
\end{algorithmic}
\end{algorithm}

\textbf{Correctness of decoding:} 
Consider any set of indices $\ca{S} \subseteq [n]$ such that $|\ca{S}|=\epsilon k$ and $\ca{S}\cap \s{supp}(\fl{x})=\emptyset$. By using the property of $(k,\epsilon k)$-list disjunct matrix $\fl{B}$, there exists an index $j \in \ca{S}$ and a row $\fl{z}$ in $\fl{B}$ such that the support of $\fl{z}$ is disjoint from support of $\fl{x}$ ($\s{supp}(\fl{x})\cap\s{supp}(\fl{z})=\phi$) and contains $j \in \ca{S}$. Therefore, there must exist $R'$ corresponding rows in $\fl{A}$ (recall the construction of $\fl{A}$) denoted by $\fl{z}^1,\fl{z}^2,\dots,\fl{z}^{R'}$ (parameterized by $a_1,a_2,\dots,a_{R'}$ respectively and have the same support as that of $\fl{z}$) such that the support of these rows are disjoint from $\s{supp}(\fl{x})$ and contains $j \in \ca{S}$. Note that in Algorithm  \ref{algo:supp_superset3}, if
\begin{align*}
\s{sign}(\langle \fl{z}^i, \fl{x} \rangle) = 0 \quad \text{for all } i\in [R']. 
\end{align*}
then we will infer the entire support of $\fl{z}$ to be disjoint from the support of $\fl{x}$ and delete those indices. The inference is correct if $\s{supp}(\fl{z})\cap \s{supp}(\fl{x}) = \phi$ and hence $\s{supp}(\fl{z}^i)\cap \s{supp}(\fl{x}) = \phi$ for all $i\in [R']$. Therefore, by our previous argument, for any set $\ca{S}\subseteq [n]: |\ca{S}|=\epsilon k, \ca{S}\cap  \s{supp}(\fl{x})=\phi$, we will delete at least one index $j \in \ca{S}$. At the end of the algorithm, we return the surviving indices.

On the other hand, we claim that we will never delete any index that lies in $\s{supp}(\fl{x})$. Notice that for all $i \in [R']$ 
\begin{align*}
    \langle \fl{x}, \fl{z}^i\rangle = \sum_{t \in \s{supp}(\fl{x}) \cap \s{supp}(\fl{z})} \fl{x}_t \fl{z}^i_t.  
\end{align*}
From our design of the measurement matrix $\fl{A}$, for all $i\in [R']$, the entries of the vector $\fl{z}^i$ are powers of some positive number $a_i$. Since $\rho(\fl{x}) \le R$ from the statement of the lemma, $ \langle \fl{x}, \fl{z}^i\rangle$ is the evaluation of a polynomial (of degree at most $n-1$ and having at most $2R=R'-1$ sign changes) at the number $a_i$. Clearly, if $|\s{supp}(\fl{x}) \cap \s{supp}(\fl{z})|=0$ then $\langle \fl{x}, \fl{z}^i\rangle =0$ for all $i \in [R']$. On the other hand, if $|\s{supp}(\fl{x}) \cap \s{supp}(\fl{z})| \neq 0$, then $\langle \fl{x}, \fl{z}^i\rangle \neq 0$ for all $i \in [R']$. This is because the polynomial 
\begin{align*}
    p(r) = \sum_{t \in (\s{supp}(\fl{x}) \cap \s{supp}(\fl{z})} \fl{x}_t r^{\alpha_t}.
\end{align*}
with at most $2R$ sign changes can have at most $2R$ positive roots (using Descartes' rule of signs); hence all of $a_1,a_2,\dots,a_{R'}$ cannot be roots of $p(r)$ as they are distinct positive numbers. Therefore, we will delete all indices that belongs to $[n]\setminus \s{supp}(\fl{x})$ except at-most $\epsilon k$ of them. This completes the proof of the theorem.
\end{proof}

\subsection{Lower Bounds}

In this section, we show  lower bounds on the  necessary number of measurements for universal $\epsilon$-approximate support recovery and universal $\epsilon$-superset recovery. 

\begin{thm}\label{thm:lower_bound}
Let $\fl{A}\in \bb{R}^{m \times n}$ be a measurement matrix such that $\s{sign}(\fl{Ax^1})\neq \s{sign}(\fl{Ax^2})$ for all $\fl{x}^1,\fl{x}^2$ satisfying $\left|\left|\fl{x}^1\right|\right|_0,\left|\left|\fl{x}^2\right|\right|_0 \le k$ and $\s{supp}(\fl{x}^1) \cap \s{supp}(\fl{x}^2) \le k(1-2\epsilon)$, for some $\epsilon<1/3$. In that case, we must have $m = \Omega\Big(\frac{k}{\epsilon} \Big(\log \frac{k}{\epsilon}\Big)^{-1}\log \frac{n-k}{\epsilon k}\Big)$.
\end{thm}

\begin{proof}
Without loss of generality, we will assume that $-1\le \fl{A}_{ij} \le 1$ for $i\in[m],j \in [n]$ since scaling by a positive number does not change the measurement output.
We will prove by contradiction that $\fl{A}$ must be a $(k(1-2\epsilon),2\epsilon k)$-list disjunct matrix. Let $\ca{B}_1,\ca{B}_2,\dots,\ca{B}_n \subseteq [m]$ be defined as follows: $\ca{B}_j = \{i \in [m]\mid \fl{A}_{ij} \neq 0\}$. Since $\fl{A}$ is a not a $(k(1-2\epsilon),2\epsilon k)$-list disjunct matrix, there must exist two disjoint sets of indices $\ca{S},\ca{T}\subseteq [n]$ such that $|\ca{S}|=2\epsilon k,|\ca{T}|=k(1-2\epsilon)$ and $\ca{B}_j \subseteq \cup_{i \in \ca{T}} \ca{B}_i$ for all $j\in \ca{S}$. Let $\fl{x}^1$ be a $k$-sparse vector such that $\s{supp}(\fl{x}^1)=\ca{T}$ and further, all indices of $\fl{Ax^1}$ in $\cup_{i \in \ca{T}} \ca{B}_i$ are $\gamma$ away from $0$. Let $$\fl{x}^2= \fl{x}^1+\sum_{j \in \ca{S}}\frac{\gamma}{2\epsilon k} \fl{e}^{j}\implies \fl{A}(\fl{x}^2-\fl{x}^1) = \fl{A}\Big(\sum_{j \in \ca{S}}\frac{\gamma}{2\epsilon k} \fl{e}^{j}\Big)$$ where $\fl{e}^{i}$ is the standard basis vector with $1$ only in the $i^{\s{th}}$ position and zero everywhere else. Since $\ca{B}_j \subseteq \cup_{i \in \ca{T}} \ca{B}_i$ for all $j \in \ca{S}$ and all entries of $\fl{A}$ are in $[-1,+1]$, we must have that $\s{sign}(\fl{A}\fl{x}^1)=\s{sign}(\fl{A}\fl{x}^2)$. Note that both $\left|\left|\fl{x}^1\right|\right|_0,\left|\left|\fl{x}^2\right|\right|_0 \le k$ and therefore, this is a contradiction. Hence $\fl{A}$ must be a $(k(1-2\epsilon),2\epsilon k)$-list disjunct matrix. Combining with the statement of Lemma \ref{disjunct_exists} (note that the condition $k\ge 2\ell$ implies that $\epsilon \le 1/3$) and the fact that $k(1-2\epsilon) \ge k/3$ for $\epsilon \le 1/3$, we obtain the statement of the theorem. 
\end{proof}

\begin{coro}
Let $\fl{A}\in \bb{R}^{m \times n}$ be a measurement matrix for universal $\epsilon$-approximate support recovery of all $k$-sparse unknown vectors for $\epsilon<1/3$. In that case, it must happen that $m = \Omega\Big(\frac{k}{\epsilon} \Big(\log \frac{k}{\epsilon}\Big)^{-1}\log \frac{n-k}{\epsilon k}\Big)$.
\end{coro}

\begin{proof}
From Theorem \ref{thm:lower_bound}, we obtained that if $m = o\Big(\frac{k}{\epsilon} \Big(\log \frac{k}{\epsilon}\Big)^{-1}\log \frac{n-k}{\epsilon k}\Big)$, then there exists $\fl{x}^1,\fl{x}^2$ satisfying $\left|\left|\fl{x}^1\right|\right|_0,\left|\left|\fl{x}^2\right|\right|_0 \le k$ and $\s{supp}(\fl{x}^1) \cap \s{supp}(\fl{x}^2) \le k(1-2\epsilon)$ such that $\s{sign}(\fl{A}\fl{x}^1)=\s{sign}(\fl{A}\fl{x}^2)$. In that case, any algorithm will not be able to distinguish between the support of $\fl{x}^1,\fl{x}^2$ which are $2\epsilon k$ apart in Hamming distance. This is a contradiction to the fact that $\fl{A}$ can be used for universal $\epsilon$-approximate recovery of all $k$-sparse unknown vectors thus proving the corollary.
\end{proof}

\begin{coro}
Let $\fl{A}\in \bb{R}^{m \times n}$ be a measurement matrix for universal $\epsilon$-superset recovery of all $k$-sparse unknown vectors for $\epsilon<1/3$. In that case, it must happen that $m = \Omega\Big(\frac{k}{\epsilon} \Big(\log \frac{k}{\epsilon}\Big)^{-1}\log \frac{n-k}{\epsilon k}\Big)$.
\end{coro}

\begin{proof}

From Proposition \ref{prop:compare}, we known that $\epsilon$-superset recovery is a strictly harder objective than $\epsilon$-approximate support recovery.
Therefore the lower bound in Theorem \ref{thm:lower_bound} extends to this setting as well.
\end{proof}

\section{Open Questions}\label{sec:limit}
Since there is a gap by a factor of $\sqrt{k}$ between the upper and lower bounds for measurement complexity in superset recovery, it will be interesting to obtain either a matching lower bound  or improve our upper bound further to match the linear lower bound. We conjecture the later to be the case, and it will be possible by clever design of polynomials with additional properties for the measurements. It will also be interesting to figure out the limits of using binary measurement matrices for support recovery. 

We are exploring if our results on universal superset recovery or approximate support recovery can be used for improving state of the art measurement complexities \cite{JLBB13} in approximately recovering the unknown signal vector itself. From a practical perspective, it would be interesting to obtain results which are robust to the assumption that the unknown signal vector is sparse; in other words, even if the signal vector has a tail, the designed algorithm can still recover the $k$ entries having the largest magnitude.

\paragraph{Acknowledgement:} This research is supported in part by NSF awards CCF 2133484, CCF 2127929, and CCF 1934846. 



\bibliographystyle{abbrv}

\appendix

\section{Proof of Theorem \ref{thm:second}}\label{sec:appendix}


 Our sensing matrix will be denoted by $\fl{A}\in \bb{R}^{m \times n}$ where $m$ is going to be determined later. Each entry of the matrix $\fl{A}$ is sampled independently according to $\ca{N}(0,1)$ ( Gaussian distribution with zero mean and variance one.) The $m$ measurements (rows of $\fl{A}$) must distinguish between vectors whose supports have a pairwise intersection of size at most $(1-2\epsilon) k$ and satisfy the dynamic range being bounded from above by $\eta$ since otherwise, the recovery algorithm cannot return a single set that is simultaneously an $\epsilon$-approximate support for both vectors. In order to prove our theorem, we will directly use the following result from \cite{flodin2019superset} showing a useful property of random Gaussian measurements: 
 
 \begin{lemma}[Lemma 16 in \cite{flodin2019superset}]\label{lem:sep}
Let $\fl{x}$ and $\fl{y}$ be two unit vectors in $\bb{R}^n$ with $\left|\left|\fl{x}-\fl{y}\right|\right|_2>\gamma$, and take $\fl{h}\in \bb{R}^n$ to be a random vector with entries drawn i.i.d according to $\ca{N}(0,1)$. Let $B_{\delta}(\fl{x})=\{\f{p}\in \bb{R}^n:\left|\left|\fl{p}-\fl{x}\right|\right|_2\le \delta\}$ be the ball of radius $\delta$ centered around $\fl{x}$. Then, we must have that 
\begin{align*}
    \Pr(\forall \fl{p}\in B_{\delta}(\fl{x}),\forall \fl{q}\in B_{\delta}(\fl{y}), \s{sign}(\fl{h}^{T}\fl{x})\neq \s{sign}(\fl{h}^{T}\fl{y})) \ge \frac{\gamma-2\delta \sqrt{n}}{\pi}.
\end{align*}
\end{lemma}



\begin{algorithm}[htbp]
\caption{\textsc{Approximate Support  Recovery}($\epsilon$) \label{algo:bd_dynamic}}
\begin{algorithmic}[1]
\REQUIRE $\eta$, $\fl{y}=\s{sign}(\fl{A x})$ where every entry of $\fl{A}$ is sampled according to $\ca{N}(0,1)$.

\STATE Compute $\fl{\hat{x}}$ to be the solution of  
\begin{align*}
    \min \left|\left|\fl{x}\right|\right|_0 \quad \text{subject to }\fl{Ax}= \fl{y} \; \text{and} \; \kappa(\fl{x})\le \eta.
\end{align*}
\STATE Return $\s{supp}(\fl{\hat{x}})$.
\end{algorithmic}
\end{algorithm}


 The probability that the $m$ measurements (rows of $\fl{A}$) are not able to distinguish between a fixed pair of $k$-sparse vectors separated by $\gamma$ in euclidean distance is at most 
\begin{align*}
    \Big(1-\frac{\gamma-2\delta \sqrt{2k}}{\pi}\Big)^{m}
\end{align*}
where we used the fact that the union of support of two $k$-sparse vectors has size at most $2k$. Consider two $k$-sparse signal vectors $\fl{x},\fl{y} \in \bb{R}^{n}$ satisfying $\kappa(\fl{x}),\kappa(\fl{y})\le \eta$ for some known $\eta>1$ such that $\left|\s{supp}(\fl{x}) \cap \s{supp}(\fl{y})\right| \le k(1-2\epsilon)$ for $\epsilon\ge 1/2k$. 
Let $\ca{S}_1 \triangleq \s{supp}(\fl{x})\setminus \s{supp}(\fl{y})$ and $\ca{S}_2 \triangleq \s{supp}(\fl{y})\setminus \s{supp}(\fl{x})$. Again note that
\begin{align*}
    &\min_{i \in [n]:\fl{u}_i \neq 0} \left|\fl{u}_i\right| \ge \frac{1}{\eta} \cdot \max_{i \in [n]:\fl{u}_i \neq 0} \left|\fl{u}_i\right| \quad \text{ if } \kappa(\fl{u}) \le \eta \\
    &\implies \min_{i \in [n]:\fl{u}_i \neq 0} \fl{u}_i^2 \ge \frac{1}{k\eta^2} \cdot \sum_{i \in [n]:\fl{u}_i \neq 0} \fl{u}_i^2  \quad \text{ if } \kappa(\fl{u}) \le \eta \\
    &\implies \min_{i \in [n]:\fl{u}_i \neq 0} \left|\fl{u}_i\right| \ge \frac{1}{\eta\sqrt{k}} \cdot \left|\left| \fl{u}\right|\right|_2  \quad \text{ if } \kappa(\fl{u}) \le \eta
\end{align*}

In that case, it must happen that 

\begin{align*}
    \left|\left|\frac{\fl{x}}{\left|\left|\fl{x}\right|\right|_2}- \frac{\fl{y}}{\left|\left|\fl{y}\right|\right|_2}\right|\right|_2 \ge \frac{\left|\left|\fl{x}_{\mid\ca{S}_1}\right|\right|_2}{\left|\left|\fl{x}\right|\right|_2}+\frac{\left|\left|\fl{y}_{\mid\ca{S}_2}\right|\right|_2}{\left|\left|\fl{y}\right|\right|_2} \ge \frac{2\sqrt{\epsilon}}{\eta}.
\end{align*}
Now, following \cite{flodin2019superset}, we can construct a $\delta$-cover $\ca{S}$ of all $k$-sparse unit vectors which is known to exist with ${n \choose k}(3/\delta)^{k}$ points.
Let $\fl{x}'\in \ca{S}$ and $\fl{y}'\in \ca{S}$ be the nearest vectors in the $\delta$-cover to $\fl{x}$ and $\fl{y}$ respectively. By using triangle inequality, we will have that $||\fl{x}'-\fl{y}'||_2 \ge 2\sqrt{\epsilon}\eta^{-1} -2\delta$. Hence, it is sufficient for the sensing matrix $\fl{A}$ to distinguish between pairs of distinct vectors $\fl{u}',\fl{v}' \in \ca{S}$ such that $||\fl{u}'-\fl{v}'||_2 \ge 2\sqrt{\epsilon}\eta^{-1} -2\delta$. Therefore, we substitute $\gamma=2\sqrt{\epsilon}/\eta, \delta=\gamma/3(1+\sqrt{2k})$ and by taking a union bound over all pairs of vectors $\fl{u}',\fl{v}' \in \ca{S}$ such that $||\fl{u}'-\fl{v}'||_2 \ge 2\sqrt{\epsilon}\eta^{-1} -2\delta$, we can bound the probability of error in decoding from above as:
\begin{align*}
   \Pr(\text{Error in Decoding}) \le {n \choose k}^{2}\Big(\frac{3}{\delta}\Big)^{2k}\Big(1-\frac{\gamma-2\delta(1+ \sqrt{2k})}{\pi}\Big)^{m}
\end{align*}

If the probability of error is less than $1$, then there exists a measurement matrix that is able to recover an $\epsilon$-approximate support for all $k$-sparse unknown vectors whose dynamic range is bounded from above by $\eta$. Hence, we have

\begin{align*}
    &{n \choose k}^{2}\Big(\frac{3}{\delta}\Big)^{2k}\Big(1-\frac{\gamma}{3\pi}\Big)^{m} 
    \le {n \choose k}^{2}\Big(\frac{10\sqrt{2k}}{\gamma}\Big)^{2k}\exp\Big(-\frac{m\gamma}{3\pi}\Big) <1 \\
    &\implies 2k\log \frac{en}{k}+2k\log\frac{10\sqrt{2k}}{\gamma} -\frac{m\gamma}{3\pi} <0 \\
    &\implies m \ge \frac{3\pi k\eta }{2\sqrt{\epsilon}} \log \frac{5en\eta}{\sqrt{\epsilon}} \\
    & \implies m \ge \frac{6\pi k\eta }{2\sqrt{\epsilon}} \log 5en\eta
\end{align*}
where in the last step, we used the fact that $\epsilon \ge 1/2k$.
 Hence, we get that there exists a matrix $\fl{A}$ with $m=O(k\eta\epsilon^{-1/2} \log n\eta)$ measurements such that we will have $\s{sign}(\fl{Ax}) \neq \s{sign}(\fl{Ay})$ for any two  $k$-sparse vectors $\fl{x},\fl{y} \in \bb{R}^{n}$ satisfying $\left|\s{supp}(\fl{x}) \cap \s{supp}(\fl{y})\right| \le k(1-\epsilon)$ and $\kappa(\fl{x}),\kappa(\fl{y})\le \eta$. This completes the proof of the theorem.

\end{document}